\documentclass[a4paper,technote]{IEEEtran}
\usepackage{etex}
\usepackage{graphicx}
\usepackage{amssymb,amsthm}
\usepackage[cmex10]{amsmath}
\usepackage[noadjust]{cite}
\usepackage{subcaption}
\newtheorem{observation}{Observation}

\title{A Short Note on the Jensen-Shannon Divergence between Simple Mixture Distributions}

\author{Bernhard C. Geiger
\thanks{Bernhard C. Geiger (geiger@ieee.org) is with the
 Know-Center GmbH, Graz, Austria.}
\thanks{The work of Bernhard C. Geiger has partly been funded by the Erwin Schr\"odinger Fellowship J 3765 of the Austrian Science Fund. The Know-Center is funded within the Austrian COMET Program - Competence Centers for Excellent Technologies - under the auspices of the Austrian Federal Ministry of Transport, Innovation and Technology, the Austrian Federal Ministry of Digital and Economic Affairs, and by the State of Styria. COMET is managed by the Austrian Research Promotion Agency FFG.}
}

\newcommand{\dom}[1]{\mathcal{#1}}
\newcommand{\pt}{\tilde{p}}
\newcommand{\kld}[2]{\mathbb{D}_{KL}(#1\Vert #2)}
\newcommand{\jsd}[3]{\mathbb{D}_{JS,#3}(#1,#2)}
\newcommand{\sjsd}[2]{\mathbb{D}_{JS}(#1,#2)}
\newcommand{\ent}[1]{H\left(#1\right)}
\newcommand{\supp}[1]{\mathrm{supp}(#1)}
\newcommand{\diff}{\mathrm{d}}

\begin{document}
\maketitle

\begin{abstract}
This short note presents results about the symmetric Jensen-Shannon divergence between two discrete mixture distributions $p_1$ and $p_2$. Specifically, for $i=1,2$, $p_i$ is the mixture of a common distribution $q$ and a distribution $\pt_i$ with mixture proportion $\lambda_i$. In general, $\pt_1\neq \pt_2$ and $\lambda_1\neq\lambda_2$. We provide experimental and theoretical insight to the behavior of the symmetric Jensen-Shannon divergence between $p_1$ and $p_2$ as the mixture proportions or the divergence between $\pt_1$ and $\pt_2$ change. We also provide insight into scenarios where the supports of the distributions $\pt_1$, $\pt_2$, and $q$ do not coincide.
\end{abstract}

\section*{Motivation}
Suppose there are three types of dice (red, blue, and green), each of which is characterized by a specific probability distribution over its faces. Suppose further that there are two urns (A and B), of which A contains only red and green dice, while B contains only blue and green dice. The proportion of red dice in A and of blue dice in B shall be known. Suppose finally that a player randomly chooses one of the urns, picks a die from the chosen urn at random, and rolls the die. You only observe the outcome of the die roll. In the problem of guessing, given only the observed die roll, the urn from which the die was chosen, the Jensen-Shannon (JS) divergence plays a significant role in bounds on the probability of guessing correctly. In this short note we evaluate the JS divergence as a function of the probability distribution over the faces of the three types of dice, and as a function of the proportion of dice of a given color in the respective urn.

\section{Notation and Assumptions}
We consider probability mass functions (PMFs) on a common finite alphabet $\dom{X}$, i.e., all PMFs in this work are $\dom{X}\to[0,1]$. Let $\supp{\cdot}$ denote the support of a PMF; e.g., for PMF $r$, $\supp{r}:=\{x\in\dom{X}{:}\ r(x)>0\}\subseteq\dom{X}$.

We denote the entropy of a discrete random variable (RV) with PMF $r$ and the Kullback-Leibler divergence between two PMFs $r_1$ and $r_2$ as
\begin{subequations}
\begin{equation}
 \ent{r} := -\sum_{x\in\supp{r}} r(x) \log r(x)
\end{equation}
and~\cite[p.~18]{Cover_Information}
\begin{equation}
 \kld{r_1}{r_2} := \sum_{x\in\supp{r_1}} r_1(x) \log \frac{r_1(x)}{r_2(x)}
\end{equation}
respectively, where $\log$ denotes the natural logarithm. The JS divergence between two PMFs $r_1$ and $r_2$ and with a weight $0\le \pi\le 1$ is defined as~\cite[eq.~(4.1)]{Lin_JSDivergence}
\begin{align}\label{eq:JSD}
 \jsd{r_1}{r_2}{\pi} &:= \pi \kld{r_1}{r_M} + (1-\pi)\kld{r_2}{r_M}\notag\\
 &=\ent{r_M} - \pi\ent{r_1} - (1-\pi)\ent{r_2}
\end{align}
where
\begin{equation}
 r_M := \pi r_1 + (1-\pi) r_2.
\end{equation}
For the sake of simplicity, we assume that $\dom{X}=\supp{r_M}$.
\end{subequations}
JS divergence admits an operational characterization in binary classification. Specifically, let $\pi$ be the prior probability of class 1, and let $1-\pi$ be the prior probability of class 2. Let further $r_i$ be the feature distribution under class $i$. Then, JS divergence appears in upper and lower bounds on the error probability $P_e$ of feature-based classification, i.e.,~\cite[Th.~4~\&~5]{Lin_JSDivergence}
\begin{equation*}
 \frac{\left(h_2(\pi)-\jsd{r_1}{r_2}{\pi}\right)^2}{4} \le P_e \le \frac{h_2(\pi)-\jsd{r_1}{r_2}{\pi}}{2}
\end{equation*}
where $h_2(x):=-x\log x-(1-x)\log(1-x)$.

\section{JS Divergence between Simple Mixture Distributions}

\begin{figure*}[t]
 \begin{subfigure}[c]{0.3\textwidth}
  \includegraphics[width=\textwidth]{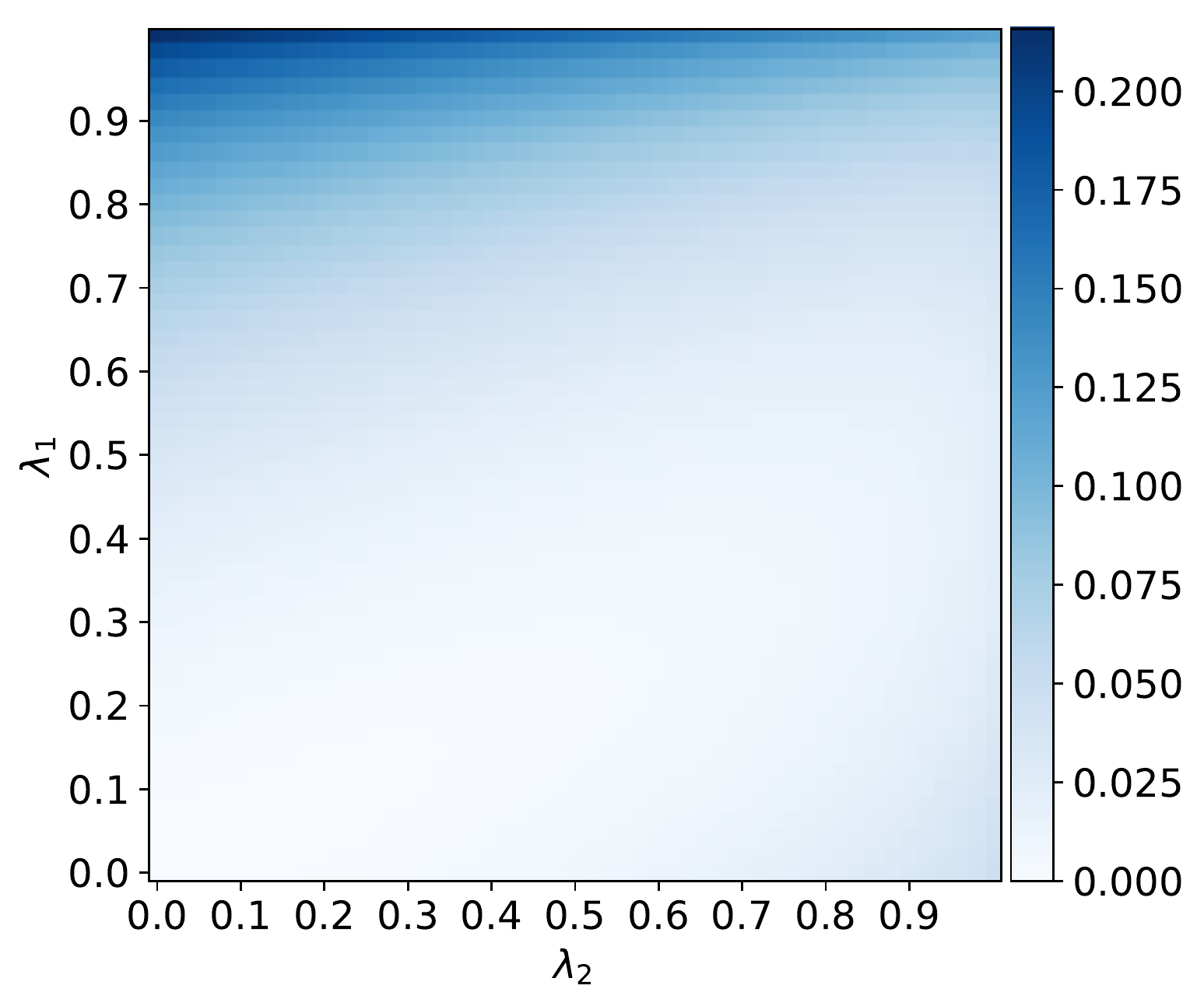}
  \caption{}
  \label{fig:gen_imshow}
 \end{subfigure}
 \hfill
 \begin{subfigure}[c]{0.3\textwidth}
 \includegraphics[width=\textwidth]{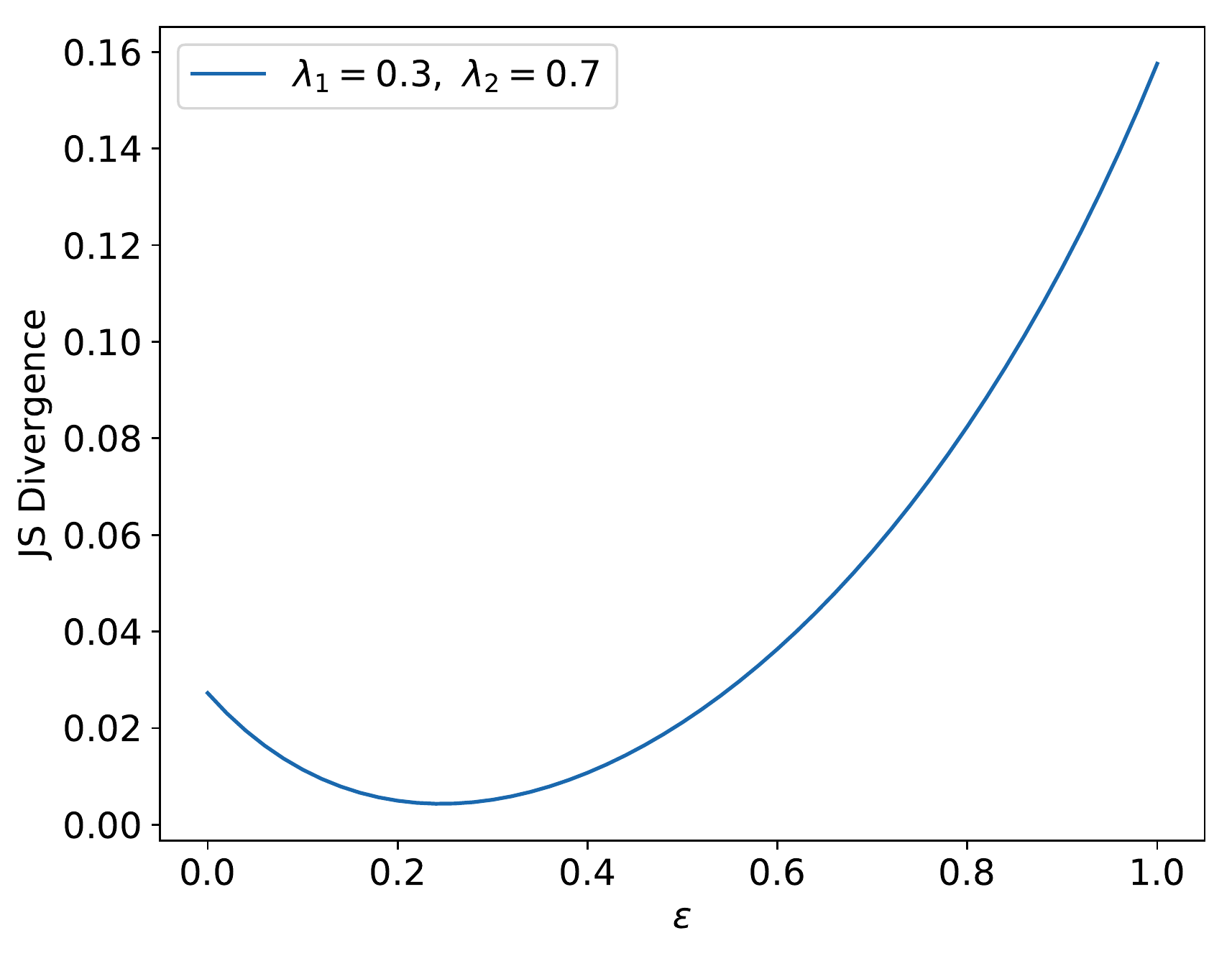}
 \caption{}
 \label{fig:error}
 \end{subfigure}\hfill
 \begin{subfigure}[c]{0.3\textwidth}
 \includegraphics[width=\textwidth]{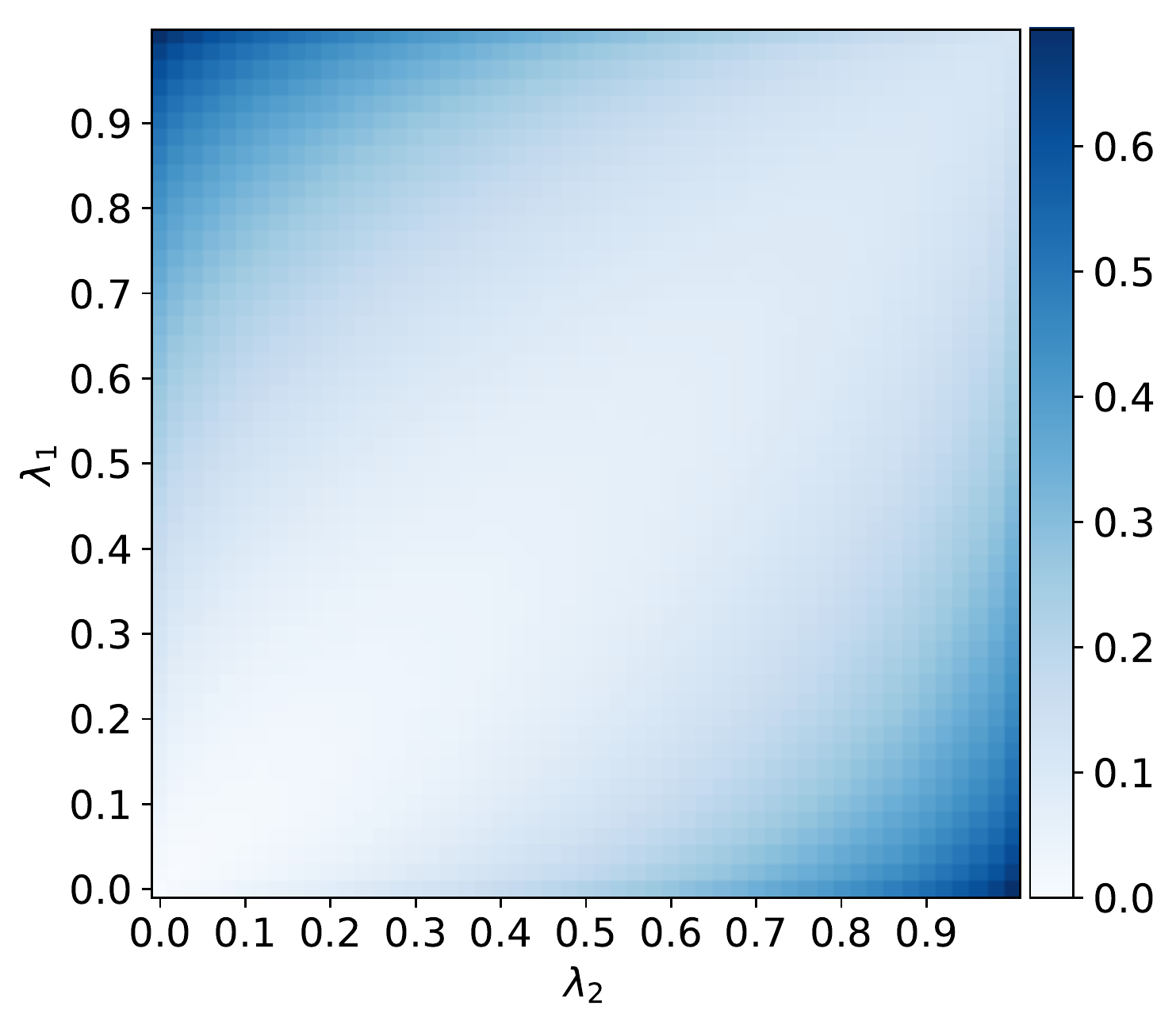}
 \caption{}
  \label{fig:PnonQ_imshow}
 \end{subfigure}\\
  \begin{subfigure}[c]{0.3\textwidth}
 \includegraphics[width=\textwidth]{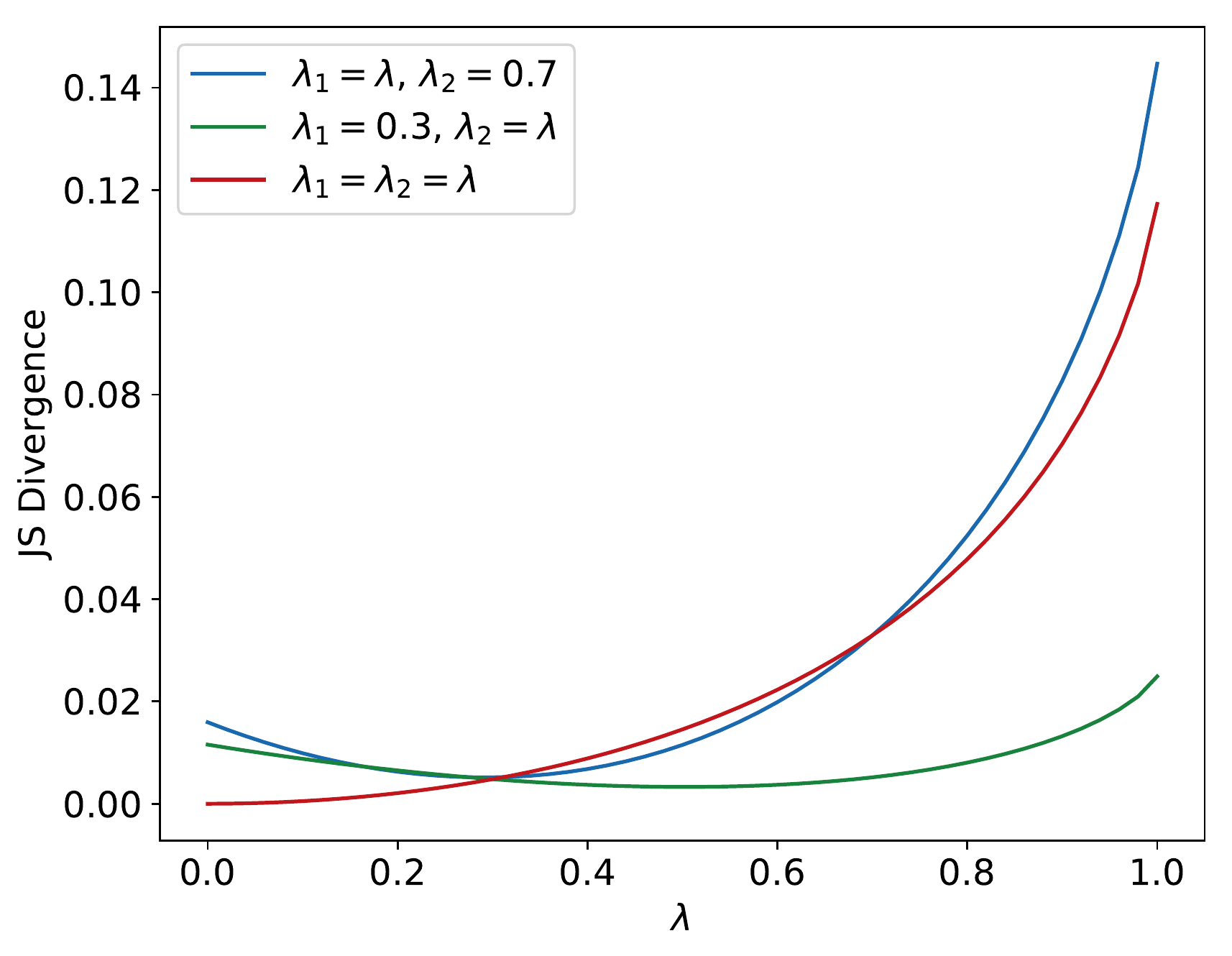}
 \caption{}
  \label{fig:gen_curves}
 \end{subfigure}\hfill
  \begin{subfigure}[c]{0.3\textwidth}
 \includegraphics[width=\textwidth]{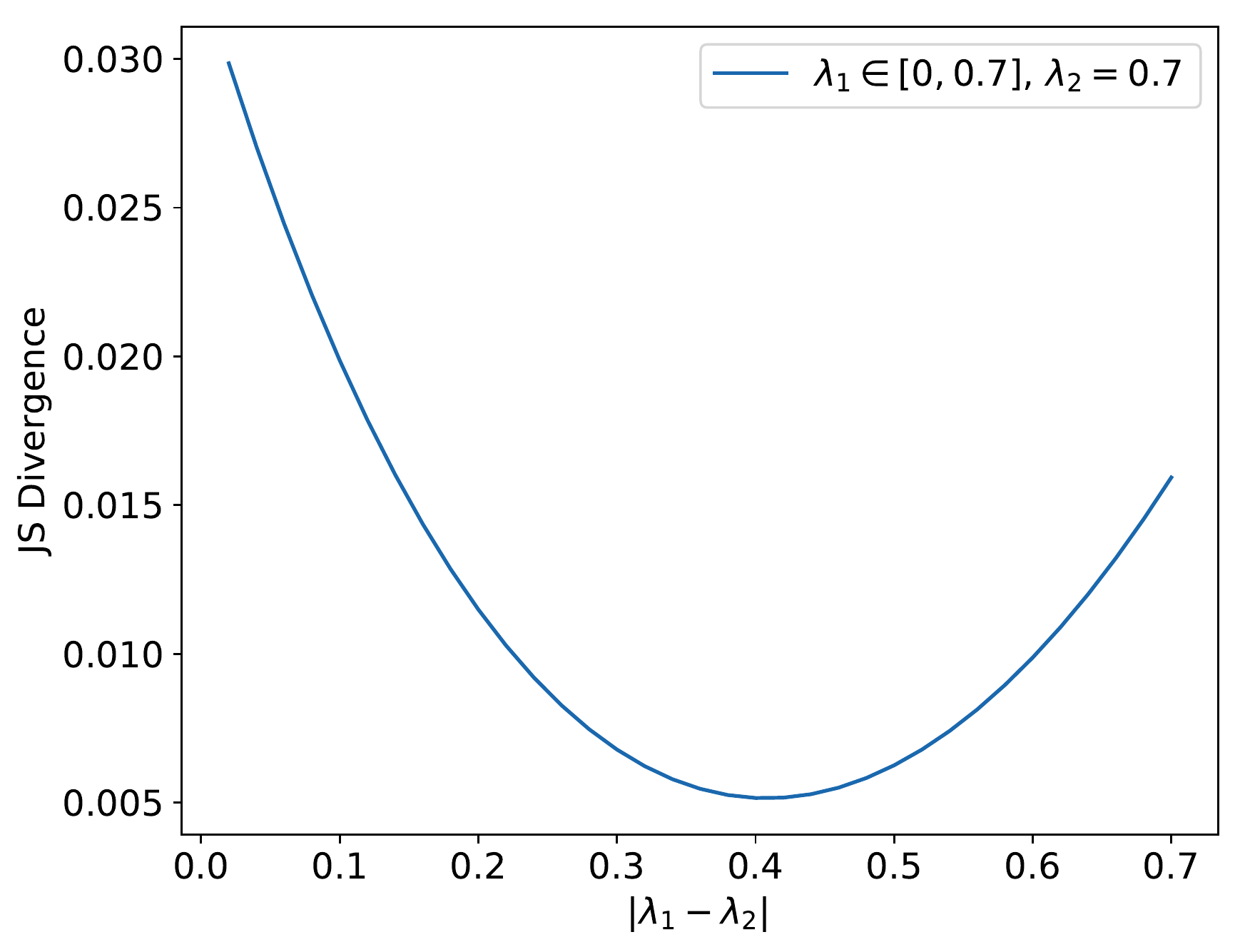}
 \caption{}
  \label{fig:relativeSize}
 \end{subfigure}\hfill
  \begin{subfigure}[c]{0.3\textwidth}
 \includegraphics[width=\textwidth]{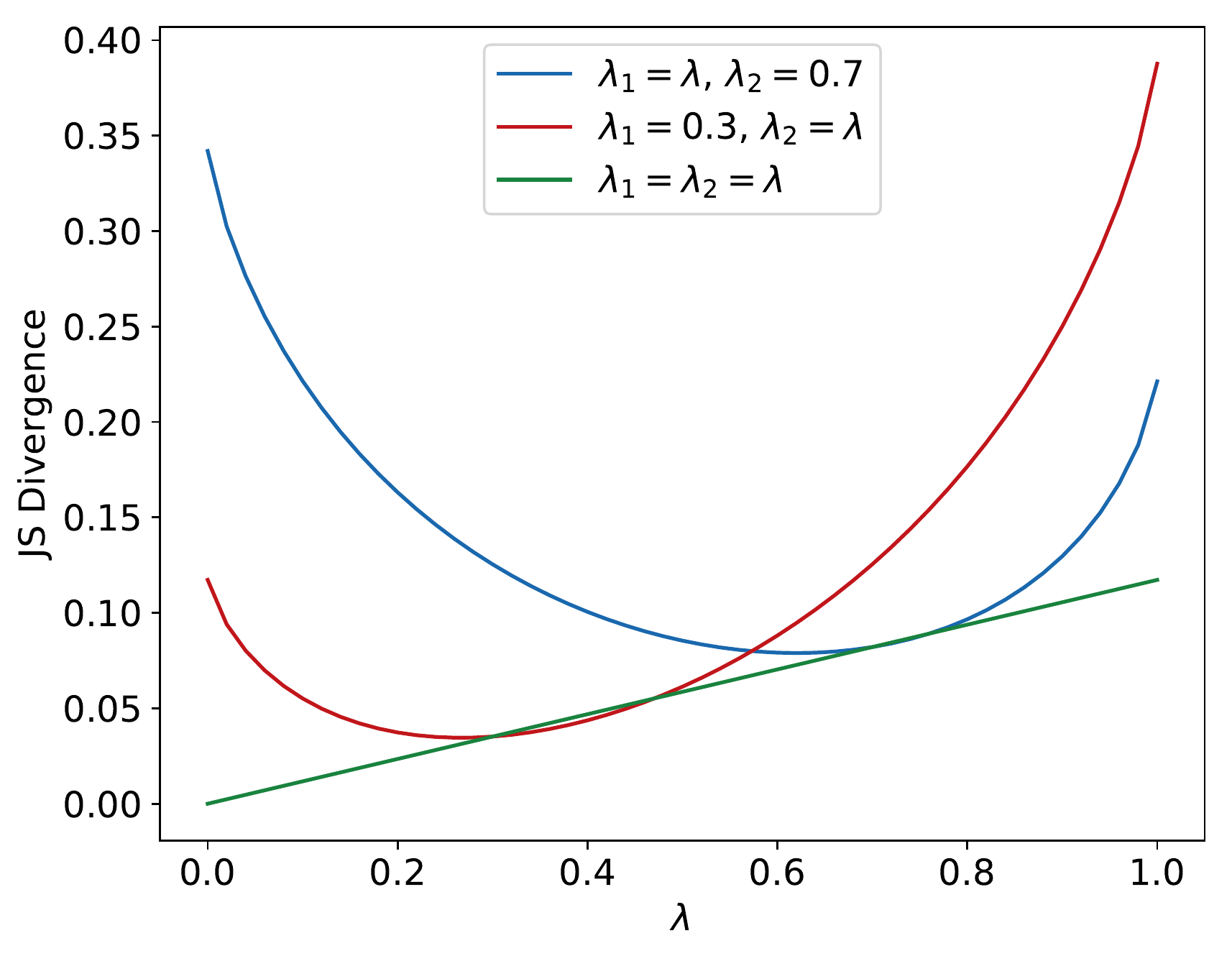}
 \caption{}
  \label{fig:PnonQ_curves}
 \end{subfigure}
 \caption{(a),(c) Symmetric JS divergence between $p_1$ and $p_2$ as a function of the mixture proportions. (d),(f) Colorplots (a),(c) evaluated at specific lines in $[0,1]^2$. (b) and (e) are obtained by evaluating the colorplot in (a) for different values of $\varepsilon$ and $|\lambda_1-\lambda_2|$, respectively. See text for details.}
\end{figure*}

Suppose that $\pt_1$, $\pt_2$, and $q$ are PMFs on the common finite alphabet $\dom{X}$ and let $\lambda_1,\lambda_2\in[0, 1]$. We define
\begin{equation}
  p_1 := \lambda_1 \pt_1 + (1-\lambda_1) q \quad \text{and}\quad
 p_2 := \lambda_2 \pt_2 + (1-\lambda_2) q
\end{equation}
i.e., the two distributions are mixtures of a common and a potentially different distribution. We consider the symmetric JS divergence between these distributions, abbreviating $\jsd{p_1}{p_2}{\frac{1}{2}}=:\sjsd{p_1}{p_2}$. Our first observation is negative.

\begin{table}[t]
 \centering
 \caption{Simulation Setting: $\dom{X}=\{1,2,3,4,5,6\}$}\label{tab:simulation}
 \begin{tabular}{c|cccccc}
  $x$ & 1 & 2 & 3 & 4 & 5 & 6\\
  \hline
  $\pt_1(x)$ & 1 & 0 & 0 & 0 & 0 & 0 \\
  $\pt_2(x)$ & $1-\varepsilon$ & $\varepsilon$ & 0 & 0 & 0 & 0 \\
  $q(x)$ & 0.5 & 0.4& 0.025& 0.025& 0.025& 0.025
 \end{tabular}
\end{table}

\begin{observation}
 $\sjsd{p_1}{p_2}$ is neither monotonic in the mixture proportions $\lambda_1$ and $\lambda_2$, nor in the difference between mixture proportions $|\lambda_1-\lambda_2|$, nor in the divergence between $\pt_1$ and $\pt_2$. 
\end{observation}

While obviously $\sjsd{p_1}{p_2}=0$ if $\lambda_1=\lambda_2=0$, it is not true that increasing the mixture proportion of one distribution relative to the other increases the symmetric JS divergence. An intuitive explanation for this is that if $\pt_1$ and $\pt_2$ are similar, then the mixture proportions $\lambda_1$ and $\lambda_2$ should also be similar to minimize the symmetric JS divergence. 

Consider for example the setting in Table~\ref{tab:simulation} with $\varepsilon=0.3$. Fig.~\ref{fig:gen_imshow} shows the JS divergence as a function of the mixture proportions. Evaluating this plot at specific lines yields Fig.~\ref{fig:gen_curves}. One can see that, fixing $\lambda_1=0.3$, $\sjsd{p_1}{p_2}$ achieves its minimum for $\lambda_2\approx 0.5$. Similarly, for a fixed $\lambda_2=0.7$, the JS divergence is minimized for $\lambda_1\approx 0.3$. Consequently, for these settings $\sjsd{p_1}{p_2}$ is not monotonic in the mixture proportions.

Similar considerations hold for the difference between mixture proportions and the divergence between $\pt_1$ and $\pt_2$. For example, apropriately setting $\lambda_1$ and $\lambda_2$ can compensate the effect of $\pt_1$ and $\pt_2$ being unequal. This situation is depicted in Fig.~\ref{fig:relativeSize} where Fig.~\ref{fig:gen_imshow} is evaluated at $\lambda_2=0.7$ and at various values of $\lambda_1\in[0,0.7]$. It can be seen that the optimal value of $\lambda_1$ in this range is not $0.7$ (for which $|\lambda_1-\lambda_2|=0$), but close to $0.3$. 

In analogy, $\pt_1$ and $\pt_2$ being different can compensate the effect of a difference between $\lambda_1$ and $\lambda_2$. Figure~\ref{fig:error} displays this behavior for the example in Table~\ref{tab:simulation}. As it can be seen, for $(\lambda_1,\lambda_2)=(0.3,0.7)$, the symmetric JS divergence is not minimized for $\varepsilon=0$ in which case $\pt_1=\pt_2$, but for $\varepsilon\approx 0.2$.

We next present a few positive results regarding the monotonicity of $\sjsd{p_1}{p_2}$ and the behavior of $\sjsd{p_1}{p_2}$ in case the supports of $\tilde{p}_i$ and $q$ are disjoint. The proofs are deferred to Section~\ref{sec:proofs}.

Inspecting Fig.~\ref{fig:gen_imshow} suggests that $\sjsd{p_1}{p_2}$ increases as $(\lambda_1,\lambda_2)$ increase jointly along a line through the origin. This behavior is also displayed in Figs.~\ref{fig:gen_curves} and~\ref{fig:PnonQ_curves}. We establish it as a general fact in the following observation.

\begin{observation}\label{obs:increase}
 If $(\lambda_1,\lambda_2)=(\lambda,\alpha\lambda)$ for some $\alpha>0$, then $\sjsd{p_1}{p_2}$ increases monotonically with $\lambda$.
\end{observation}

Intuitively, if $\pt_1=\pt_2$, then the symmetric JS divergence between $p_1$ and $p_2$ should increase if the difference between mixture proportions increases.

\begin{observation}\label{obs:samedist}
 If $\pt_1=\pt_2$, then $\sjsd{p_1}{p_2}$ increases monotonically with $|\lambda_1-\lambda_2|$.
\end{observation}

We finally evaluate the scenario where the supports of $\pt_1$ and $\pt_2$ are disjoint from the support of $q$. Suppose we draw a sample of either $p_1$ or $p_2$ and it is our task to determine from which distribution it was drawn. We assume to know the supports of $\pt_1$, $\pt_2$, and $q$. If the drawn sample is from the support of $q$, then there is now way to distinguish between $p_1$ and $p_2$; one only has the chance to distinguish $p_1$ from $p_2$ if the drawn sample is from the supports of $\pt_1$ and $\pt_2$. The following observation shows that in this case the JS divergence does not depend on the PMF $q$.

\begin{observation}\label{obs:disjointsupports}
If $\supp{\pt_i}\cap \supp{q}=\emptyset$ for $i=1,2$, then
\begin{multline}\label{eq:PnonQ}
 \sjsd{p_1}{p_2} = \sjsd{(\lambda_1,1-\lambda_1)}{(\lambda_2,1-\lambda_2)}\\
 + \frac{\lambda_1+\lambda_2}{2} \jsd{\pt_1}{\pt_2}{\frac{\lambda_1}{\lambda_1+\lambda_2}}.
\end{multline}
\end{observation}

An example for this scenario is depicted in Figs.~\ref{fig:PnonQ_imshow} and~\ref{fig:PnonQ_curves}. The simulation setting coincides with the one in Table~\ref{tab:simulation}, with $\varepsilon=0.3$ and $q$ replaced by a uniform distribution on $\{3,4,5,6\}$. Obviously, since the supports of $\pt_i$ and $q$, $i=1,2$ are disjoint, $\sjsd{p_1}{p_2}$ achieves its maximum at $(\lambda_1,\lambda_2)=(0,1)$ and $(\lambda_1,\lambda_2)=(1,0)$.

Now suppose $\pt_1=\pt_2$. In this case, $p_1$ and $p_2$ can only be distinguished if they mix $\pt_1$ and $\pt_2$ with different proportions (since otherwise $p_1=p_2$). This is accounted for in the first term of~\eqref{eq:PnonQ}. Next, suppose that $\lambda_1=\lambda_2=\lambda$. The only chance to distinguish $p_1$ and $p_2$ results from $\pt_1$ and $\pt_2$ being different, which is accounted for in the second term of~\eqref{eq:PnonQ}. Moreover, in this case $p_1$ and $p_2$ can be distinguished more easily if the common mixture proportion is large. Since finally in this case we have $\sjsd{p_1}{p_2}=\lambda\sjsd{\pt_1}{\pt_2}$, one can see that $\sjsd{p_1}{p_2}$ increases linearly with $\lambda$ if $\lambda_1=\lambda_2=\lambda$ (see Fig.~\ref{fig:PnonQ_curves}). In Fig.~\ref{fig:PnonQ_curves} one can moreover see that, in the more general setting of $\lambda_1\neq\lambda_2$ and $\pt_1\neq \pt_2$, a monotonic increase is not observable. This is due to the nontrivial interplay between the two terms in~\eqref{eq:PnonQ}.

\section{Proofs}\label{sec:proofs}
Suppose that $r(\cdot,\lambda)$ is a PMF parameterized by $\lambda$. Simple calculus shows that
\begin{equation}\label{eq:derivativeEntropy}
 \frac{\diff}{\diff \lambda} \ent{r(\cdot,\lambda)} = -\sum_{x\in\supp{r}}  \frac{\diff}{\diff \lambda}r(x,\lambda) \left[1+\log r(x,\lambda) \right].
\end{equation}
Let further
\begin{equation}\label{eq:mixturedistribution}
 p_M := \frac{1}{2}p_1 + \frac{1}{2}p_2 = \frac{\lambda_1}{2}\pt_1 +\frac{\lambda_2}{2}\pt_2 + \frac{2-\lambda_1-\lambda_2}{2} q.
\end{equation}

\begin{proof}[Proof of Observation~\ref{obs:increase}]
With~\eqref{eq:JSD} and~\eqref{eq:derivativeEntropy} we obtain
\begin{align}
 &-\frac{\diff}{\diff\lambda} \sjsd{p_1}{p_2}\notag\\
 &=\sum_{x\in\dom{X}}\left[\frac{\pt_1(x)}{2}+\frac{\alpha\pt_2(x)}{2}-\frac{1+\alpha}{2}q(x)\right] \left(1+\log p_M(x)\right)\notag\\
 &\quad\quad - \frac{1}{2}\sum_{x\in\supp{p_1}} \left[\pt_1(x)-q(x)\right] \left(1+\log p_1(x)\right)\notag\\
 &\quad\quad - \frac{1}{2}\sum_{x\in\supp{p_2}} \left[\alpha\pt_2(x)-\alpha q(x)\right] \left(1+\log p_2(x)\right)\notag\\
 &=\sum_{x\in\dom{X}}\left[\frac{\pt_1(x)}{2}+\frac{\alpha\pt_2(x)}{2}-\frac{1+\alpha}{2}q(x)\right] \log p_M(x)\notag\\
 &\quad\quad - \frac{1}{2} \left[\pt_1(x)-q(x)\right] \log p_1(x)\notag\\
 &\quad\quad - \frac{1}{2} \left[\alpha\pt_2(x)-\alpha q(x)\right]\log p_2(x)\label{eq:obs_increase:line1}
\end{align}
where the last equality is obtained by continuous extension of the logarithm to obtain $0\log 0:=0$. Observe further that
\begin{align*}
 &\left[\alpha\pt_2(x)-\alpha q(x)\right]\log p_2(x)\\ 
 &=  \left[\alpha\pt_2(x)-\alpha q(x)\right]\log\left(\alpha\lambda\pt_2(x)+(1-\alpha\lambda)q(x)\right)\\ 
 &=  \left[\alpha\pt_2(x)-\alpha q(x)\right]\log\left(\lambda(\alpha\pt_2(x)-\alpha q(x))+q(x)\right)
\end{align*}
i.e., it is of shape $r\log(\lambda r+q)$. It can be shown that this function is convex in $r$ if $r$, $q$, and $\lambda$ are nonnegative. It thus follows that each summand in~\eqref{eq:obs_increase:line1} is nonpositive. This completes the proof.
\end{proof}

\begin{proof}[Proof of Observation~\ref{obs:samedist}]
We first assume that $\lambda_1,\lambda_2\in(0,1)$. We treat the remaining cases separately. Let w.l.o.g.\ $\lambda:=\min\{\lambda_1, \lambda_2\}$ and $\Delta\lambda:=|\lambda_2-\lambda_1|$. We obtain that
 \begin{multline}
  \sjsd{p_1}{p_2} = \ent{\frac{2\lambda+\Delta\lambda}{2}\pt+\frac{2-2\lambda-\Delta\lambda}{2}q}\\
  -\frac{1}{2}\ent{\lambda\pt+(1-\lambda)q}-\frac{1}{2}\ent{(\lambda+\Delta\lambda)\pt+(1-\lambda-\Delta\lambda)q}
 \end{multline}
and thus 
 \begin{align}
  &\frac{\diff}{\diff\Delta\lambda} \sjsd{p_1}{p_2}\notag\\
  &= \sum_{x\in\dom{X}}\left[\frac{\pt(x)-q(x)}{2}\right]\notag\\
  &\quad\quad\times\log\left(\frac{(\lambda+\Delta\lambda)\pt(x)+(1-\lambda-\Delta\lambda)q(x)}{(\lambda+\frac{\Delta\lambda}{2})\pt(x)+(1-\lambda-\frac{\Delta\lambda}{2})q(x)}\right)\label{eq:samedist:line1}\\
  &= \sum_{x\in\dom{X}}\left[\frac{\pt(x)-q(x)}{2}\right]\notag\\
  &\quad\quad\times\log\left(1+\frac{\Delta\lambda}{2}\frac{\pt(x)-q(x)}{(\lambda+\frac{\Delta\lambda}{2})\pt(x)+(1-\lambda-\frac{\Delta\lambda}{2})q(x)}\right)\label{eq:samedist:line2}\\
  &= \sum_{x\in\dom{X}}\left[\frac{\pt(x)-q(x)}{2}\right]\log\left(1+\frac{\Delta\lambda}{2}\frac{\pt(x)-q(x)}{p_M(x)}\right)\label{eq:samedist:line4}.
 \end{align}  
Note that the argument of the logarithm in~\eqref{eq:samedist:line4} is equivalent to the argument of the logarithm in~\eqref{eq:samedist:line1}, which is positive for $x\in\dom{X}$ if $\lambda_1,\lambda_2\in(0,1)$. We can therefore break the sum into to parts, one in which $\pt\ge q$ and one in which $\pt<q$, and bound the logarithm by $x-1 \ge \log x \ge \frac{x-1}{x}$. Thus we obtain
\begin{align}
&\frac{\diff}{\diff\Delta\lambda} \sjsd{p_1}{p_2}\notag\\
&\ge \sum_{x{:}\ \pt(x)\ge q(x)}\left[\frac{\pt(x)-q(x)}{2}\right]^2 \frac{\frac{\Delta\lambda}{p_M(x)}}{1+\frac{\Delta\lambda}{2}\frac{\pt(x)-q(x)}{p_M(x)}}\notag\\
&\quad\quad + \sum_{x{:}\ \pt(x)< q(x)} \left[\frac{\pt(x)-q(x)}{2}\right]^2 \frac{\Delta\lambda}{p_M(x)}\\
&\ge 0
\end{align}
where the last line follows since each summand is nonnegative.

If $(\lambda_1,\lambda_2)=(0,0)$, then $\sjsd{p_1}{p_2}=0$ and~\eqref{eq:PnonQ} holds. If $(\lambda_1,\lambda_2)=(1,1)$, then~\eqref{eq:PnonQ} holds trivially. If $(\lambda_1,\lambda_2)=(1,0)$, then the supports of $p_1$ and $p_2$ are disjoint and $\sjsd{p_1}{p_2}=\log(2)$, i.e., it achieves its maximum value. The same value is achieved by the first term in~\eqref{eq:PnonQ}, while the second term can be shown to be zero in this case. This completes the proof.
\end{proof}

\begin{proof}[Proof of Observation~\ref{obs:disjointsupports}]
We get
\begin{align*}
 &\kld{p_1}{p_M}\notag\\
&= \sum_{x\in\supp{\pt_1}} \lambda_1\pt_1(x) \log\frac{\lambda_1\pt_1(x)}{\frac{\lambda_1}{2}\pt_1(x) +\frac{\lambda_2}{2}\pt_2(x)} \notag\\
&\quad + \sum_{x\in\supp{q}} (1-\lambda_1)q(x) \log \frac{(1-\lambda_1)q(x) }{\frac{2-\lambda_1-\lambda_2}{2} q(x)}\\
&=  \lambda_1 \sum_{x\in\supp{\pt_1}}\pt_1(x) \log\frac{\pt_1(x)}{\frac{\lambda_1}{\lambda_1+\lambda_2}\pt_1(x) +\frac{\lambda_2}{\lambda_1+\lambda_2}\pt_2(x)} \notag\\
&\quad + (1-\lambda_1) \log \frac{(1-\lambda_1) }{\frac{2-\lambda_1-\lambda_2}{2}} + \lambda_1\log\frac{\lambda_1}{\frac{\lambda_1+\lambda_2}{2}}\\
&= \lambda_1 \kld{\pt_1}{\pt} + \kld{(\lambda_1,1-\lambda_1)}{\lambda,1-\lambda)}
\end{align*}
where $\pt\triangleq\frac{\lambda_1}{\lambda_1+\lambda_2}\pt_1 +\frac{\lambda_2}{\lambda_1+\lambda_2}\pt_2$ and where $\lambda\triangleq\frac{\lambda_1+\lambda_2}{2}$. The proof is completed by repeating the same steps for $\kld{p_2}{p_M}$ and inserting it in~\eqref{eq:JSD}.
\end{proof}

\section*{Acknowledgments}
The author thanks Roman Kern, Know-Center GmbH and Institute for Interactive Systems and Data Science, Graz University of Technology, for fruitful discussions.

\bibliographystyle{IEEEtran}
\bibliography{IEEEabrv,references}

\end{document}